\documentclass[runningheads]{llncs}
\usepackage{todonotes}
\usepackage[end]{algpseudocode}
\usepackage{xspace,xargs}
\usepackage[utf8]{inputenc}

\usepackage[utf8]{inputenc} 
\usepackage[T1]{fontenc}    
\usepackage{hyperref}       
\usepackage{url}            
\usepackage{booktabs}       
\usepackage{amsfonts}       
\usepackage{nicefrac}       
\usepackage{microtype}      
\usepackage{xspace}
\usepackage{amsmath}
\usepackage{blindtext}
\usepackage{arydshln}
\usepackage{framed}

\long\def\ignore#1{\vskip 0pt}

\newcommand{\SA}{\ensuremath{\mathsf{SA}}\xspace}
\newcommand{\C}{\ensuremath{\mathsf{C}}\xspace}

\newcommand{\LA}{\ensuremath{\mathsf{LA}}\xspace}

\newcommand{\ISA}{\ensuremath{\mathsf{ISA}}\xspace}
\newcommand{\A}{\ensuremath{\mathsf{A}}\xspace}
\newcommand{\NSVa}[1]{\ensuremath{\mathsf{NSV_{#1}}}\xspace}
\newcommand{\NSV}{\ensuremath{\mathsf{NSV}}\xspace}

\newcommand{\Next}{\ensuremath{\mathsf{NEXT}}\xspace}
\newcommand{\Prev}{\ensuremath{\mathsf{PREV}}\xspace}

\renewcommand{\A}{\ensuremath{\mathsf{A}}\xspace}

\newcommand{\etal}{{\it et al.}\xspace}

\newcommand{\avelyn}{\mathsf{avelyn}}

\newcommand{\tbf}{\textbf}

\newcommand{\sacak}{\mbox{\textsc{SACA-K}}\xspace}
\newcommand{\nsv}{\mbox{\textsc{NSV-Lyndon}}\xspace}
\newcommand{\baier}{\mbox{\textsc{Baier-LA}}\xspace}
\newcommand{\baiersa}{\mbox{\textsc{Baier-LA+SA}}\xspace}
\newcommand{\louza}{\mbox{\textsc{BWT-Lyndon}}\xspace}

\newcommand{\ourA}{\mbox{\textsc{SACA-K+LA}}\xspace}
\newcommand{\ourB}{\mbox{\textsc{SACA-K+LA}-17n}\xspace}
\newcommand{\ourC}{\mbox{\textsc{SACA-K+LA}-13n}\xspace}
\newcommand{\ourD}{\mbox{\textsc{SACA-K+LA}-9n}\xspace}

\begin{document}
%
\title{Inducing the Lyndon Array}
%
%
\author{
Felipe A. Louza\inst{1} 
\and
Sabrina~Mantaci\inst{2} 
\and
Giovanni~Manzini\inst{3,4} 
\and\\
Marinella~Sciortino\inst{2} 
\and
Guilherme~P.~Telles\inst{5} 
}

\authorrunning{F. A. Louza et al.}
%
\institute{
Department of Computing and Mathematics, University of S\~ao Paulo, Brazil\\
\email{louza@usp.br}\\
\and
Dipartimento di Matematica e Informatica, University of Palermo, {Italy}\\
\email{\{sabrina.mantaci,marinella.sciortino\}@unipa.it}\\
\and
University of Eastern Piedmont, {Alessandria, Italy} \\
\and
IIT CNR, Pisa, Italy\\
\email{giovanni.manzini@uniupo.it}\\
\and
Institute of Computing, University of Campinas, Brazil\\
\email{gpt@ic.unicamp.br}
}
\maketitle              
\begin{abstract}

In this paper we propose a variant of the {induced suffix sorting} algorithm by Nong (TOIS, 2013) that computes simultaneously the Lyndon array and the suffix array of a text in {$O(n)$} time using {$\sigma + O(1)$ words} of working space, where $n$ is the length of the text and $\sigma$ is the alphabet size.
Our result improves the previous best space requirement for linear time computation of the Lyndon array. In fact, all the known linear algorithms for Lyndon array computation use suffix sorting {as a preprocessing step} and use {$O(n)$ words of working space in addition to the Lyndon array and suffix array}. 
Experimental results with real and synthetic datasets show that our algorithm is not only space-efficient but also fast in practice.

\keywords{Lyndon array, Suffix array, induced suffix sorting, lightweight algorithms}

\end{abstract}

\section{Introduction}\label{s:intro}

The suffix array is a central data structure for string processing.
Induced suffix sorting is a remarkably powerful technique for the construction of the suffix array. Induced sorting was introduced by Itoh and Tanaka~\cite{it99} and later refined by Ko and Aluru~\cite{KoAlu03} and by Nong \etal~\cite{NongZC09,Nong2011}. 
In 2013, Nong~\cite{tois/Nong13} proposed a space efficient linear time algorithm based on induced sorting, called \sacak, which uses only $\sigma + O(1)$ words of working space, where $\sigma$ is the alphabet size and the working space is the space used in addition to the input and the output. 
Since a small working space is a very desirable feature, there have been many algorithms adapting induced suffix sorting to the computation of data structures related to the suffix array, such as the Burrows-Wheeler transform~\cite{Okanohara2009}, the $\Phi$-array~\cite{Goto2014}, the LCP array~\cite{Fischer2011,ipl/LouzaGT17}, and the document array~\cite{tcs/LouzaGT17}.

The Lyndon array of a string is a powerful tool that generalizes the idea of Lyndon factorization.
In the Lyndon array ($\LA$) of string $T=T[1]\ldots T[n]$ over the alphabet $\Sigma$, each entry $\LA[i]$, with $1\leq i\leq n$, stores the length of the longest Lyndon factor of $T$ starting at that position $i$.
Bannai \etal~\cite{siamcomp/BannaiIINTT17} used Lyndon arrays to prove the conjecture by Kolpakov and Kucherov~\cite{focs/KolpakovK99} that the number of runs (maximal periodicities) in a string of length $n$ is smaller than $n$. In~\cite{tcs/CrochemoreR2018} the authors have shown that the computation of the Lyndon array of $T$ is strictly related to the construction of the Lyndon tree~\cite{Hohlweg2003} of the string $\$T$ (where the symbol $\$$ is smaller than any symbol of the alphabet $\Sigma$).

In this paper we address the problem of designing a space economical linear time algorithm for the computation of the Lyndon array.  As described in~\cite{Franek2016,jda/LouzaSMT18}, there are several algorithms to compute the Lyndon array. It is noteworthy that the ones that run in linear time (cf.~\cite{Baier2016,tcs/CrochemoreR2018,Franek2016,FranekPS17,jda/LouzaSMT18}) use the sorting of the suffixes (or a partial sorting of suffixes) of the input string as a preprocessing step. 
Among the linear time algorithms, the most space economical {is the one in~\cite{Franek2016}} which, in addition to the $n \log \sigma$ bits for the input string plus $2n$ words for the Lyndon array and suffix array, uses a stack whose size depends on the structure of the input. Such stack is relatively small for non pathological texts, but in the worst case its size can be up to $n$ words. 
Therefore, the overall space in the worst case can be up to $n \log \sigma$ bits plus $3n$ words.

{In this paper we propose a variant of the algorithm \sacak that computes in linear time the Lyndon array as a by-product of suffix array construction. Our algorithm uses overall $n \log \sigma$ bits plus $2n+\sigma + O(1)$ words of space. 
This bound makes our algorithm the one with the best worst case space bound among the linear time algorithms. Note that the $\sigma + O(1)$ words of working space of our algorithm is optimal for strings from alphabets of constant size. 
Our experiments show that our algorithm is competitive in practice compared to the other linear time solutions to compute the Lyndon array.}

\section{Background}\label{s:background}

Let $T=T[1]\dots T[n]$ be a string of length $n$ over a fixed ordered alphabet $\Sigma$ of size $\sigma$, where $T[i]$ denotes the $i$-th symbol of $T$.
We denote $T[i,j]$ as the factor of $T$ starting from the $i$-th symbol and ending at the $j$-th symbol. A suffix of $T$ is a factor of the form $T[i,n]$ and is also denoted as $T_i$. In the following we assume that any integer array of length $n$ with values in the range $[1,n]$ takes {$n$ words ($n \log n$ bits)} of space.

Given $T=T[1]\dots T[n]$, the {\em $i$-th rotation} of $T$ begins with $T[i+1]$, corresponding to the string $T'=T[i+1]\dots T[n]T[1]\dots T[i]$.
Note that, a string of length $n$ has $n$ possible rotations. A string $T$ is a {\em repetition} if there exists a string $S$ and an integer $k>1$ such that $T=S^k$, otherwise it is called  {\em primitive}. If a string is primitive, all of its rotations are different.

A primitive string $T$ is called a {\em Lyndon word} if it is the lexicographical least among its rotations. 
For instance, the string $T=abanba$ is not a Lyndon word, while it is its rotation $aabanb$ is. A \emph{Lyndon factor} of a string $T$ is a factor of $T$ that is a Lyndon word. For instance, $anb$ is a Lyndon factor of $T=abanba$.

\begin{definition}
Given a string $T=T[1]\dots T[n]$, the Lyndon array (LA) of $T$ is an array of integers in the range $[1,n]$ that, at each position $i=1,\dots,n$, stores the length of the longest Lyndon factor of $T$ starting at $i$:
$$
\LA[i] = \max\{\ell~|~T[i,i+\ell-1] \mbox{ is a Lyndon word}\}.
$$
\end{definition}

The suffix array (\SA)~\cite{MM93} of a string $T=T[1]\dots T[n]$ is an array of integers in the range $[1,n]$ that gives the lexicographic order of all suffixes of $T$, that is $T_{\SA[1]}<T_{\SA[2]}<\dots<T_{\SA[n]}$.
The inverse suffix array (\ISA) stores the inverse permutation of \SA, such that $\ISA[\SA[i]]=i$.
The suffix array can be computed in $O(n)$ time using $\sigma + O(1)$ words of working space~\cite{tois/Nong13}. 

Usually when dealing with suffix arrays it is convenient to append to the string $T$ a special end-marker symbol $\$$ (called sentinel) that does not occur elsewhere in {$T$ and $\$$ is smaller than any other symbol in $\Sigma$.}
Here we assume that $T[n]=\$$. Note that the values $\LA[i]$, for $1\leq i\leq n-1$ do not change when the symbol $\$$ is appended at the position $n$. 
{Also, string $T=T[1]\dots T[n-1]\$$ is always primitive.}

Given an array of integers $\A$ of size $n$, the next smaller value (\NSV) array of $\A$, denoted $\NSVa{A}$, is an array of size $n$ such that $\NSVa{A}[i]$ contains the smallest position $j>i$ such that $\A[j]<\A[i]$, or $n+1$ if such a position $j$ does not exist. Formally:
$$ 
\NSVa{A}[i]=\min\bigl\{\{n+1\}\cup\{i<j\leq n \mid \A[j]<\A[i]\}\bigr\}.
$$
As an example, in Figure~\ref{f:la} we consider the string $T=banaananaanana\$$, and its Suffix Array (\SA), Inverse Suffix Array (\ISA), Next Smaller Value array of the \ISA (\NSVa{\ISA}), and Lyndon Array (\LA). 
We also show all the Lyndon factors starting at each position of $T$. 

\begin{figure}[t]
    \centering
    \includegraphics[width=.75\textwidth]{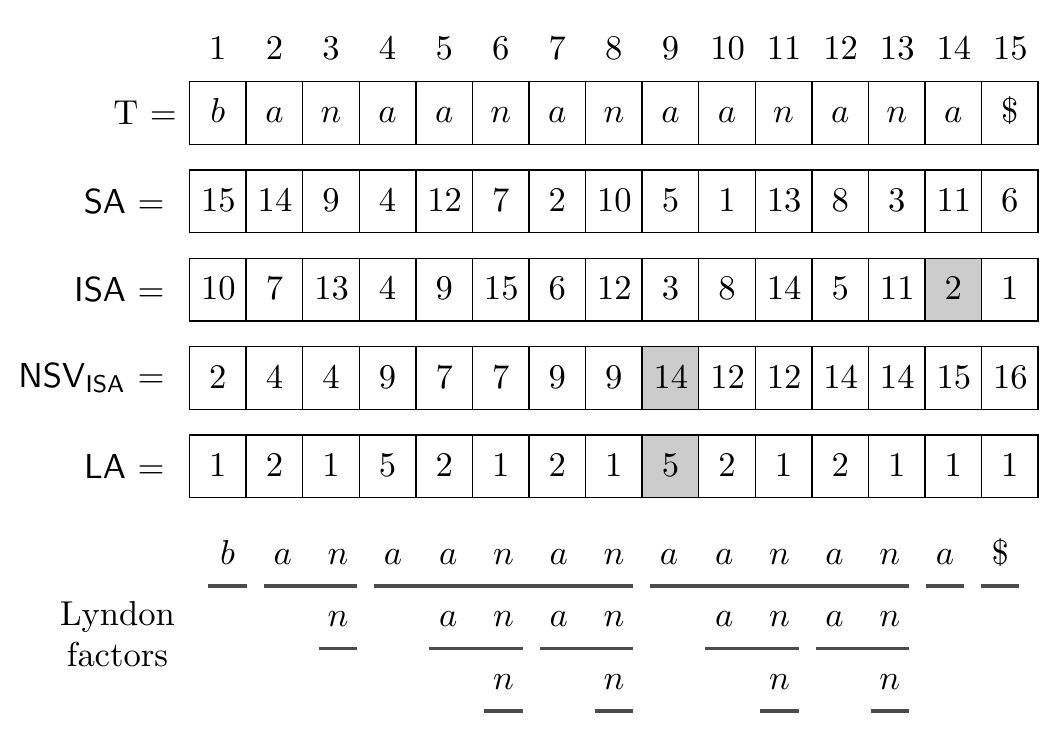}
    \caption{\SA, \ISA, \NSVa{\ISA}, \LA and all Lyndon factors for $T=banaananaanana\$$}
    \label{f:la}
\end{figure}

If the $\SA$ of $T$ is known, the Lyndon array $\LA$ can be computed in linear time thanks to the following lemma that rephrases a result in \cite{Hohlweg2003}:

\begin{lemma}\label{l:sa_lyndon}
The factor $T[i, i+ \ell-1]$ is the longest Lyndon factor of $T$ starting at $i$ iff $T_{i}<T_{i+k}$, for $1\leq k<\ell$, and $T_{i}>T_{i+\ell}$.
Therefore, $\LA[i]=\ell$.\qed
\end{lemma}

Lemma~\ref{l:sa_lyndon} can be reformulated in terms of the inverse suffix array~\cite{Franek2016}, such that
$\LA[i]=\ell$ iff $\ISA[i]<\ISA[i+k]$, for $1\leq k <\ell$, and $\ISA[i]>\ISA[i+\ell]$. In other words, $i+\ell = \NSV_{\ISA}[i]$.
Since given {\ISA} we can compute $\NSVa{\ISA}$ in linear time using an auxiliary stack~\cite{Goto2013,Ohlebusch2013} of size $O(n)$ words, we can then derive \LA, {in the same space of $\NSVa{\ISA}$}, in linear time using the formula:
\begin{equation}\label{e:nsv_lyndon}
    \LA[i] = \NSV_{\ISA}[i]-i\mbox{, for }1 \leq i \leq n.
\end{equation}
Overall, this approach uses $n \log \sigma$ bits for $T$ plus $2n$ words for \LA and \ISA, and the space for the auxiliary stack.

Alternatively, \LA can be computed in linear time from the Cartesian tree~\cite{cacm/Vuillemin80} built for \ISA~\cite{tcs/CrochemoreR2018}.
Recently, Franek \etal~\cite{FranekPS17} {observed} that \LA can be computed in linear time during the suffix array construction algorithm by Baier~\cite{Baier2016} using overall $n \log \sigma$ bits plus $2n$ words for \LA and \SA plus $2n$ words for auxiliary integer arrays.
Finally, Louza \etal~\cite{jda/LouzaSMT18} introduced an algorithm that computes \LA in linear time during the Burrows-Wheeler inversion, using 
$n \log \sigma$ bits for $T$ plus $2n$ words for \LA and an auxiliary integer array, plus a stack with twice the size as the one used to compute $\NSVa{\ISA}$ (see Section~\ref{s:experiments}).

Summing up, the most economical linear time solution for computing the Lyndon array is the one based on~\eqref{e:nsv_lyndon} that requires, in addition to $T$ and \LA, $n$ words of working space plus an auxiliary stack. The stack size is small for non pathological inputs but can use $n$ words in the worst case (see also Section~\ref{s:experiments}).
Therefore, considering only \LA as output, the working space is $2n$ words in the worst case.

\subsection{Induced Suffix Sorting}\label{s:sacak}

The algorithm \sacak~\cite{tois/Nong13} uses a technique called induced suffix sorting to compute \SA in linear time using only $\sigma + O(1)$ words of working space. In this technique each suffix $T_i$ of $T[1,n]$ is classified according to its lexicographical rank relative to $T_{i+1}$.

\begin{definition}
A suffix $T_i$ is S-type if $T_i<T_{i+1}$, otherwise $T_i$ is L-type.
We define $T_n$ as S-type.
A suffix $T_i$ is LMS-type (leftmost S-type) if $T_i$ is S-type and $T_{i-1}$ is L-type.
\end{definition}

The type of each suffix can be computed with a right-to-left  scanning of $T$~\cite{Nong2011}, or otherwise it can be computed on-the-fly in constant time during Nong's algorithm~\cite[Section 3]{tois/Nong13}. 
By extension, the type of each symbol in $T$ can be classified according to the type of the suffix starting with such symbol. In particular $T[i]$ is LMS-type if and only if $T_i$ is LMS-type.

\begin{definition}
An LMS-factor of $T$ is a factor that begins with a LMS-type symbol and ends with the following LMS-type symbol. 
\end{definition}

We remark that LMS-factors do not establish a factorization of $T$ since each of them overlaps with the following one by one symbol. 
By convention, $T[n,n]$ is always an LMS-factor.  
The LMS-factors of $T=banaananaanana\$$ are shown in Figure~\ref{f:sacak}, where the type of each symbol is also reported. The LMS types are the grey entries. 
Notice that in \SA all suffixes starting with the same symbol $c\in \Sigma$ can be partitioned into a $c$-bucket.
We will keep an integer array $\C[1,\sigma]$ where $\C[c]$ gives either the first (head) or last (tail) available position of the $c$-bucket.
Then, whenever we insert a value into the head (or tail) of a $c$-bucket, we increase (or decrease) $\C[c]$ by one.
An important remark is that within each $c$-bucket S-type suffixes are larger than L-type suffixes.
Figure~\ref{f:sacak} shows a running example of algorithm \sacak for $T=banaananaanana\$$.

\begin{figure}[t!]
    \centering
    \includegraphics[width=.95\textwidth]{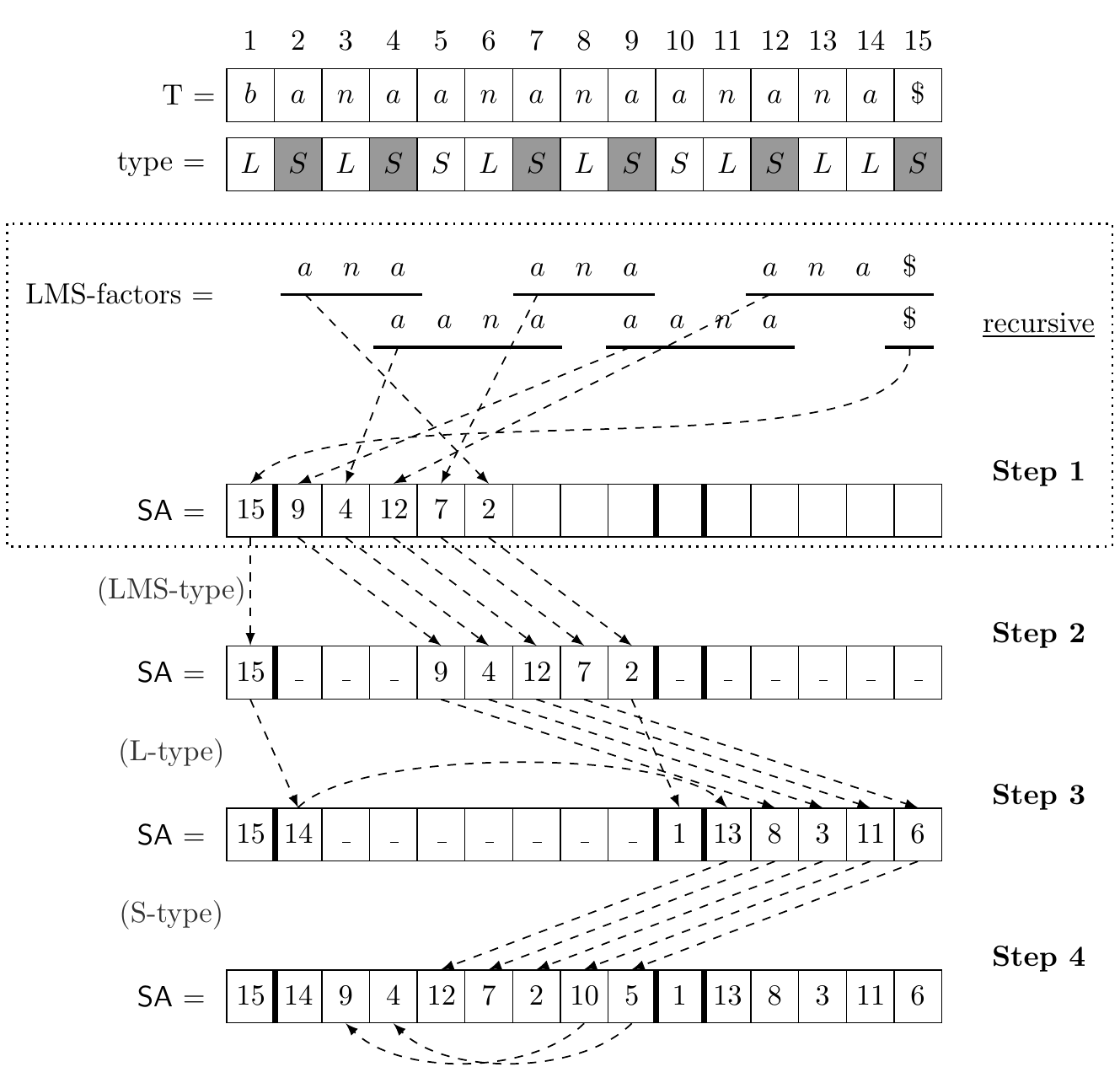}
    \caption{Induced suffix sorting steps (\sacak) for $T=banaananaanana\$$}
    \label{f:sacak}
\end{figure}

Given all LMS-type suffixes of $T[1,n]$, the suffix array can be computed as follows:

\begin{framed}{Steps:}
\begin{enumerate}
    \item Sort all LMS-type suffixes recursively into $\SA^1$, stored in $\SA[1,n/2]$.
    \item Scan $\SA^1$ from right-to-left, and insert the LMS-suffixes into the tail of their corresponding $c$-buckets in \SA.
    \item Induce L-type suffixes by scanning \SA left-to-right: for each suffix $\SA[i]$, if $T_{\SA[i]-1}$ is L-type, insert $\SA[i]-1$ into the head of its bucket.
    \item Induce S-type suffixes by scanning \SA right-to-left: for each suffix $\SA[i]$, if $T_{\SA[i]-1}$ is S-type, insert $\SA[i]-1$ into the tail of its bucket.    
\end{enumerate}
\end{framed}

{Step $1$ considers the string $T^1$ obtained by concatenating the lexicographic names of all the consecutive LMS-factors (each different string is associated with a symbol that represents its lexicographic rank). Note that $T^1$ is defined over an alphabet of size $O(n)$ and that its length is at most $n/2$. The \sacak algorithm is applied recursively to sort the suffixes of $T^1$ into $\SA^1$, which is stored in the first half of $\SA$.} Nong \etal~\cite{Nong2011} showed that sorting the suffixes of $T^1$ is equivalent to sorting the LMS-type suffixes of $T$. We will omit details of this step, since our algorithm will not modify it.

Step $2$ obtains the sorted order of all LMS-type suffixes from $\SA^1$ scanning it from right-to-left and bucket sorting then into the tail of their corresponding $c$-buckets in $\SA$.
Step $3$ induces the order of all L-type suffixes by scanning \SA from left-to-right. Whenever suffix $T_{\SA[i]-1}$ is L-type, $\SA[i]-1$ is inserted in its final (corrected) position in \SA.

Finally, Step $4$ induces the order of all S-type suffixes by scanning \SA from right-to-left. Whenever suffix $T_{\SA[i]-1}$ is S-type, $\SA[i]-1$ is inserted in its final (correct) position in \SA.

\paragraph{Theoretical costs.}
Overall, algorithm \sacak runs in linear time using only an additional array of size $\sigma + O(1)$ words to store the bucket array~\cite{tois/Nong13}.

\section{Inducing the Lyndon array}\label{s:algorithm}

In this section we show how to compute the Lyndon array (\LA) during Step $4$ of algorithm \sacak described in Section~\ref{s:sacak}.
Initially, we set all positions $\LA[i]=0$, for $1\leq i \leq n$. In Step $4$, when \SA is scanned from right-to-left, each value $\SA[i]$, corresponding to $T_{\SA[i]}$, is read in its final (correct) position $i$ in \SA.
In other words, we read the suffixes in decreasing order from $\SA[n], \SA[n-1],\dots, \SA[1]$. We now show how to compute, during iteration $i$, the value of $\LA[\SA[i]]$.

By Lemma~\ref{l:sa_lyndon}, we know that the length of the longest Lyndon factor starting at position $\SA[i]$ in $T$, that is $\LA[\SA[i]]$, is equal to $\ell$, where $T_{\SA[i]+\ell}$ is the next suffix (in text order) that is smaller than $T_{\SA[i]}$.
In this case, $T_{\SA[i]+\ell}$ will be the first suffix in $T_{\SA[i]+1},T_{\SA[i]+2}\dots, T_n$ that has not yet been read in \SA, which means that $T_{\SA[i]+\ell}<T_{\SA[i]}$.
Therefore, during Step $4$, whenever we read $\SA[i]$, 
we compute $\LA[\SA[i]]$
by scanning $\LA[\SA[i]+1,n]$ {to the right} up to the first position $\LA[\SA[i]+\ell]=0$, and we set $\LA[\SA[i]]=\ell$.

The correctness of this procedure follows from the fact that every position in $\LA[1,n]$ is initialized with zero, and if $\LA[\SA[i]+1], \LA[\SA[i]+2], \dots, \LA[\SA[i]+\ell-1]$ are no longer equal to zero, their corresponding suffixes has already been read in positions larger than $i$ in $\SA[i,n]$, and such suffixes are larger (lexicographically) than $T_{\SA[i]}$.
Then, the first position we find $\LA[\SA[i]+\ell]=0$ corresponds to a suffix $T_{\SA[i]+\ell}$ that is smaller than $T_{\SA[i]}$, which was still not read in \SA. 
Also, $T_{\SA[i]+\ell}$ is the next smaller suffix (in text order) because we read $\LA[\SA[i]+1,n]$ from left-to-right.

Figure~\ref{f:algorithm} illustrates iterations $i=15$, $9$, and $3$ of our algorithm for $T=banaananaanana\$$.
For example, at iteration $i=9$, the suffix $T_5$ is read at position $\SA[9]$, and the corresponding value $\LA[5]$ is computed by scanning $\LA[6], \LA[7], \dots, \LA[15]$ up to finding the first empty position, which occurs at $\LA[7=5+2]$. Therefore, $\LA[5]=2$.

\begin{figure}[t!]
    \centering
    \includegraphics[width=.84\textwidth]{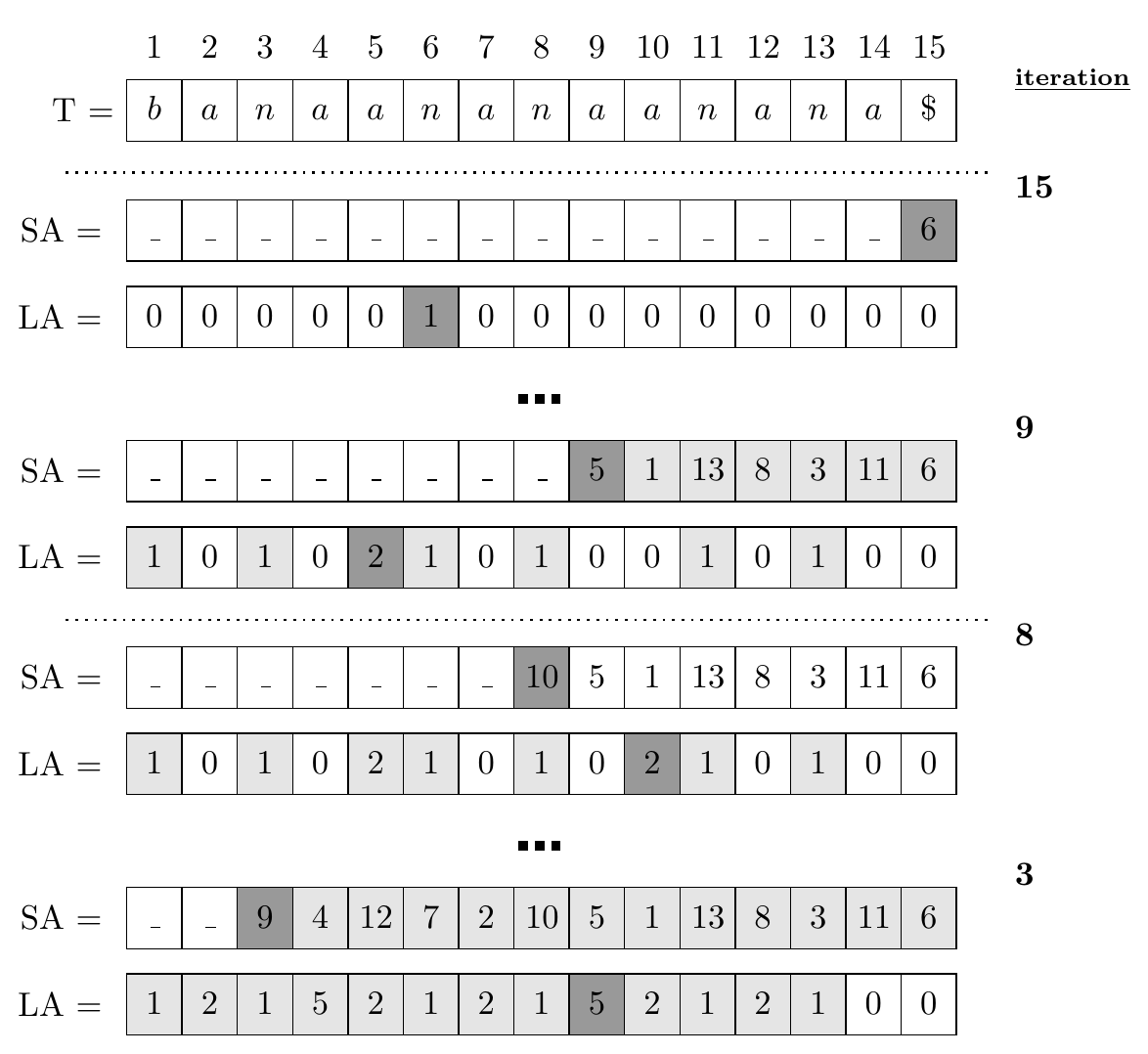}
    \caption{Running example for $T=banaananaanana\$$.}
    \label{f:algorithm}
\end{figure}

At each iteration $i=n,n-1,\dots, 1$, the value of $\LA[\SA[i]]$ is computed in additional $\LA[\SA[i]]$ steps, that is our algorithm adds $O(\LA[i])$ time for each iteration of \sacak.

Therefore, our algorithm runs in $O(n \cdot \avelyn)$ time, where $\avelyn = \sum_{i=1}^{n} \LA[i]/n$. Note that computing \LA does not need extra memory on top of the space for $\LA[1,n]$. Thus, the working space is the same as \sacak, which is $\sigma + O(1)$ words.

\begin{lemma}
The Lyndon array and the suffix array of a string $T[1,n]$ over an alphabet of size $\sigma$ can be computed simultaneously in $O(n \cdot \avelyn)$ time using $\sigma + O(1)$ words of working space,
where $\avelyn$ is equal to the average value in $\LA[1,n]$.\qed
\end{lemma}

In the next sections we show how to modify the above algorithm to reduce both its running time and its working space.

\subsection{Reducing the running time to $O(n)$}\label{sec:alt1}

We now show how to modify the above algorithm to compute each \LA entry in constant time. To this end, we store for each position $\LA[i]$ the next smaller position $\ell$ such that $\LA[\ell]=0$. We define two additional {pointer} arrays $\Next[1,n]$ and $\Prev[1,n]$:

\begin{definition}
For $i=1,\ldots,n-1$, $\Next[i] = \min\{\ell|i<\ell\leq n \mbox{ and } \LA[\ell]=0\}$. In addition, we define $\Next[n]=n+1$.
\end{definition}

\begin{definition}
For $i=2,\ldots,n$, $\Prev[i] = \ell$, such that $\Next[\ell]=i$ and $\LA[\ell]=0$. In addition, we define $\Prev[1]=0$.
\end{definition}

The above definitions depend on \LA and therefore $\Next$ and $\Prev$ are updated as we compute additional \LA entries.  
Initially, we set $\Next[i]=i+1$ and $\Prev[i]=i-1$, for $1\leq i \leq n$.
Then, at each iteration $i=n, n-1, \dots, 1$, when we compute $\LA[j]$ with $j=\SA[i]$ setting:
\begin{equation}\label{e:la_next}
  \LA[j] = \Next[j] - j  
\end{equation}
we update the {pointers arrays as follows}:
\begin{align}
\Next[\Prev[j]] & =\Next[j],\quad\mbox{ if }\Prev[j]>0\label{e:next} \\
\Prev[\Next[j]] &= \Prev[j],\quad\mbox{ if }\Next[j]<n+1 \label{e:prev}
\end{align}
The cost of computing each \LA entry is now constant, since only two additional computations (Equations~\ref{e:next} and~\ref{e:prev}) are needed. Because of the use of the arrays \Prev and \Next the working space of our algorithm is now $2n + \sigma + O(1)$ words.

\begin{theorem}
The Lyndon array and the suffix array of a string $T[1,n]$ over an alphabet of size $\sigma$ can be computed simultaneously in  $O(n)$ time using $2n + \sigma + O(1)$ words of working space.\qed
\end{theorem}

\subsection{Getting rid of a pointer array}\label{s:alt_2}

We now show how to reduce the working space of Section~\ref{sec:alt1} by storing only one array, say $\A[1,n]$, keeping $\Next/\Prev$ information together.
In a glace, we store \Next initially into the space of $\A[1,n]$, then we reuse $\A[1,n]$ to store the (useful) entries of \Prev.

Note that, whenever we write $\LA[j]=\ell$, the value in $\A[j]$, that is $\Next[j]$ is no more used by the algorithm.
Then, we can reuse $\A[j]$ to store $\Prev[j+1]$. 
Also, we know that if $\LA[j]=0$ then $\Prev[j+1]=j$.
Therefore, we can redefine \Prev in terms of \A:

\begin{equation}\label{e:a_prev}
\Prev[j]=
\begin{cases}
j-1   & \mbox{ if } \LA[j-1]=0 \\
\A[j-1] & \mbox{ otherwise}.
\end{cases}
\end{equation}

The running time of our algorithm remains the same since we have added only one extra verification to obtain $\Prev[j]$ (Equation~\ref{e:a_prev}).
Observe that whenever $\Next[j]$ is overwritten the algorithm does not need it anymore.
The working space is therefore reduced to $n + \sigma + O(1)$ words.

\begin{theorem}
The Lyndon array and the suffix array of a string $T[1,n]$ over an alphabet of size $\sigma$ can be computed simultaneously in  $O(n)$ time using $n + \sigma + O(1)$ words of working space.\qed
\end{theorem}

\subsection{Getting rid of both pointer arrays}\label{s:alt_3}

Finally, we show how to use the space of $\LA[1,n]$ to store both the auxiliary array $\A[1,n]$ and the final values of \LA. First we observe that it is easy to compute $\LA[i]$ when $T_i$ is an L-type suffix.

\begin{lemma}
$\LA[j]=1$ iff $T_{j}$ is an L-type suffix, or $i=n$.
\end{lemma}

\begin{proof}
If $T_{j}$ is an L-type suffix, then $T_{j}>T_{j+1}$ and $\LA[j]=1$.
By definition $\LA[n]=1$.\qed
\end{proof}

Notice that at Step 4 during iteration $i=n,n-1, \dots, 1$, whenever we read an S-type suffix $T_{j}$, with $j=\SA[i]$, its succeeding suffix (in text order) $T_{j+1}$ has already been read in some position in the interval $\SA[i+1,n]$ ($T_{j+1}$ have induced the order of $T_{j}$).
Therefore, the \LA-entries corresponding to S-type suffixes are always inserted on the left of a block (possibly of size one) of non-zero entries in $\LA[1,n]$.

Moreover, whenever we are computing $\LA[j]$ and we have $\Next[j]=j+k$ (stored in $\A[j]$), we know the following entries $\LA[j+1], \LA[j+2],\dots,\LA[j+k-1]$ are no longer zero, and we have to update 
$\A[j+k-1]$, corresponding to $\Prev[j+k]$ (Equation~\ref{e:a_prev}).  
In other words, we update \Prev information only for right-most entry of each block of non empty entries, which corresponds to a position of an L-type suffix because S-type are always inserted on the left of a block.

Then, at the end of the modified Step 4, if $\A[i]<i$ then $T_i$ is an L-type suffix, and we know that $\LA[i]=1$.
On the other hand, the values with $\A[i]>i$ remain equal to $\Next[i]$ at the end of the algorithm.
And we can use them to compute $\LA[i]=\A[i]-i$ (Equation~\ref{e:la_next}).

Thus, after the completion of Step 4, 
we sequentially scan $\A[1,n]$ overwriting its values with \LA as follows:
\begin{equation}\label{e:a_la}
\LA[j]=
\begin{cases}
1   & \mbox{ if } \A[j]<j \\
\A[j]-j & \mbox{ otherwise}.
\end{cases}
\end{equation}

The running time of our algorithm is still linear, since we added only a linear scan over $\A[1,n]$ {at the end of Step 4}. On the other hand, the working space is reduced to $\sigma + O(1)$ words, since we need to store only the bucket array $\C[1,\sigma]$.

\begin{theorem}\label{t:main_result}
The Lyndon array and the suffix array of a string of length $n$ over an alphabet of size $\sigma$ can be computed simultaneously in $O(n)$ time using $\sigma + O(1)$ words of working space.\qed
\end{theorem}

Note that the bounds on the working space given in the above theorems assume that the output consists of \SA and \LA. If one is interested in \LA only, then the working space of the algorithm is $n + \sigma + O(1)$ words which is still smaller that the working space of the other linear time algorithms that we discussed in Section~\ref{s:background}.

\section{Experiments}\label{s:experiments}

\begin{table}[t]
\centering
\setlength{\tabcolsep}{5.4pt} 
\renewcommand{\arraystretch}{1.2} 
\begin{tabular}{lrr|ccc||c|ccc||c}
\toprule
    &  & & \multicolumn{3}{c||}{\LA}   & \multicolumn{4}{c||}{\LA and \SA}                    & \SA     \\ 
dataset    & $\sigma$   & $n/2^{20}$   & \multicolumn{1}{c}{\rotatebox{90}{\nsv~\cite{Hohlweg2003}}}      & \multicolumn{1}{c}{\rotatebox{90}{\baier~\cite{Baier2016,FranekPS17}}}      & \multicolumn{1}{c||}{\rotatebox{90}{\louza~\cite{jda/LouzaSMT18}}}      & \multicolumn{1}{c|}{\rotatebox{90}{\baiersa~\cite{Baier2016,FranekPS17}}}      &  \multicolumn{1}{c}{\rotatebox{90}{\ourB}}        & \multicolumn{1}{c}{\rotatebox{90}{\ourC}}      & \rotatebox{90}{\ourD}      & \rotatebox{90}{\sacak~\cite{tois/Nong13}}      \\ \hline
\texttt{pitches}     & 133 & 53    & \tbf{0.15}  & 0.20       & 0.20 & 0.26        &  0.26 & 0.22 & \tbf{0.18} & 0.13 \\
\texttt{sources}     & 230 & 201   & \tbf{0.26}  & 0.28       & 0.32 & 0.37        &  0.46 & 0.41 & \tbf{0.34} & 0.24 \\
\texttt{xml}         & 97  & 282   & \tbf{0.29}  & 0.31       & 0.35 & 0.42        &  0.52 & 0.47 & \tbf{0.38} & 0.27 \\
\texttt{dna}         & 16  & 385   & 0.39        & \tbf{0.28} & 0.49 & \tbf{0.43}  &  0.69 & 0.60 & 0.52 & 0.36 \\
\texttt{english.1GB} & 239 & 1,047 & 0.46        & \tbf{0.39} & 0.56 & \tbf{0.57}  &  0.84 & 0.74 & 0.60 & 0.42 \\
\texttt{proteins}    & 27  & 1,129 & 0.44        & \tbf{0.40} & 0.53 & 0.66        &  0.89 & 0.69 & \tbf{0.58} & 0.40 \\ \hdashline
\texttt{einstein-de} & 117 & 88    & 0.34        & \tbf{0.28} & 0.38 & \tbf{0.39}  &  0.57 & 0.54 & 0.44 & 0.31 \\
\texttt{kernel}      & 160 & 246   & 0.29        & \tbf{0.29} & 0.39 & \tbf{0.38}  &  0.53 & 0.47 & \tbf{0.38} & 0.26 \\
\texttt{fib41}       & 2   & 256   & 0.34        & \tbf{0.07} & 0.45 & \tbf{0.18}  &  0.66 & 0.57 & 0.46 & 0.32 \\
\texttt{cere}        & 5   & 440   & 0.27        & \tbf{0.09} & 0.33 & \tbf{0.17}  &  0.43 & 0.41 & 0.35 & 0.25 \\ \hdashline
\texttt{bbba}         & 2   & 100   & 0.04        & \tbf{0.02} & 0.05 & \tbf{0.03}  &  0.05 & 0.04 & \tbf{0.03} & 0.03 \\ \hline
\end{tabular}
\caption{
Running time ($\mu$s/input byte).
}
\label{t:time}
\end{table}

\begin{table}[t]
\centering
\setlength{\tabcolsep}{8pt} 
\renewcommand{\arraystretch}{1.2} 
\begin{tabular}{lrr|ccc||c|ccc||c}
\toprule
    &  & & \multicolumn{3}{c||}{\LA}   & \multicolumn{4}{c||}{\LA and \SA}                    & \SA     \\ 
dataset    & $\sigma$   & $n/2^{20}$   & \multicolumn{1}{c}{\rotatebox{90}{\nsv~\cite{Hohlweg2003}}}      & \multicolumn{1}{c}{\rotatebox{90}{\baier~\cite{Baier2016,FranekPS17}}}      & \multicolumn{1}{c||}{\rotatebox{90}{\louza~\cite{jda/LouzaSMT18}}}      & \multicolumn{1}{c|}{\rotatebox{90}{\baiersa~\cite{Baier2016,FranekPS17}}}      &  \multicolumn{1}{c}{\rotatebox{90}{\ourB}}        & \multicolumn{1}{c}{\rotatebox{90}{\ourC}}      & \rotatebox{90}{\ourD}      & \rotatebox{90}{\sacak~\cite{tois/Nong13}}      \\ \hline
\texttt{pitches}     & 133 & 53    & \tbf{9}  & 17       & \tbf{9}  & 17 & 17 & 13 & \tbf{9} & 5    \\
\texttt{sources}     & 230 & 201   & \tbf{9}  & 17       & \tbf{9}  & 17 & 17 & 13 & \tbf{9} & 5 \\
\texttt{xml}         & 97  & 282   & \tbf{9}  & 17       & \tbf{9}  & 17 & 17 & 13 & \tbf{9} & 5 \\
\texttt{dna}         & 16  & 385   & \tbf{9}  & 17       & \tbf{9}  & 17 & 17 & 13 & \tbf{9} & 5 \\
\texttt{english.1GB} & 239 & 1,047 & \tbf{9}  & 17       & \tbf{9}  & 17 & 17 & 13 & \tbf{9} & 5 \\
\texttt{proteins}    & 27  & 1,129 & \tbf{9}  & 17       & \tbf{9}  & 17 & 17 & 13 & \tbf{9} & 5 \\ \hdashline
\texttt{einstein-de} & 117 & 88    & \tbf{9}  & 17       & \tbf{9}  & 17 & 17 & 13 & \tbf{9} & 5 \\
\texttt{kernel}      & 160 & 246   & \tbf{9}  & 17       & \tbf{9}  & 17 & 17 & 13 & \tbf{9} & 5 \\
\texttt{fib41}       & 2   & 256   & \tbf{9}  & 17       & \tbf{9}  & 17 & 17 & 13 & \tbf{9} & 5 \\
\texttt{cere}        & 5   & 440   & \tbf{9}  & 17       & \tbf{9}  & 17 & 17 & 13 & \tbf{9} & 5 \\ \hdashline
\texttt{bbba}         & 2   & 100   & \tbf{13} & {17} & {17} & 17 & 17 & 13 & \tbf{9} & 5 \\ \hline
\end{tabular}
\caption{
Peak space (bytes/input size).
}
\label{t:peakspace}
\end{table}

\ignore{
\begin{table}[t]
\centering
\setlength{\tabcolsep}{2.6pt} 
\renewcommand{\arraystretch}{1.2} 
\begin{tabular}{lrr|rrr||r|rrc||c}
\toprule
    &  & & \multicolumn{3}{c||}{\LA}   & \multicolumn{4}{c||}{\LA and \SA}                    & \SA     \\ 
dataset    & $\sigma$   & $n/2^{20}$   & \multicolumn{1}{c}{\rotatebox{90}{\nsv~\cite{Hohlweg2003}}}      & \multicolumn{1}{c}{\rotatebox{90}{\baier~\cite{Baier2016,FranekPS17}}}      & \multicolumn{1}{c||}{\rotatebox{90}{\louza~\cite{jda/LouzaSMT18}}}      & \multicolumn{1}{c|}{\rotatebox{90}{\baiersa~\cite{Baier2016,FranekPS17}}}      &  \multicolumn{1}{c}{\rotatebox{90}{\ourB}}        & \multicolumn{1}{c}{\rotatebox{90}{\ourC}}      & \rotatebox{90}{\ourD}      & \rotatebox{90}{\sacak~\cite{tois/Nong13}}      \\ \hline
\texttt{pitches}     & 133 & 53    & \tbf{213.0}   & 639.0     & \tbf{213.0}    & 426.0   &  426.0   & 213.0   & \tbf{1} & 1\\ 
\texttt{sources}     & 230 & 201   & \tbf{804.4}   & 2,413.2   & \tbf{804.4}    & 1,608.8 &  1,608.8 & 804.4   & \tbf{1} & 1 \\
\texttt{xml}         & 97  & 282   & \tbf{1,129.7} & 3,389.0   & \tbf{1,129.7}  & 2,259.3 &  2,259.3 & 1,129.7 & \tbf{1} & 1 \\
\texttt{dna}         & 16  & 385   & \tbf{1,540.9} & 4,622.6   & \tbf{1,540.9}  & 3,081.7 &  3,081.7 & 1,540.9 & \tbf{1} & 1 \\
\texttt{english.1GB} & 239 & 1,047 & \tbf{4,187.4} & 12,562.2  & \tbf{4,187.4}  & 8,374.8 &  8,374.8 & 4,187.4 & \tbf{1} & 1 \\
\texttt{proteins}    & 27  & 1,129 & \tbf{4,516.8} & 13,550.4  & \tbf{4,516.8}  & 9,033.6 &  9,033.6 & 4,516.8 & \tbf{1} & 1 \\ \hdashline
\texttt{einstein-de} & 117 & 88    & \tbf{353.9}   & 1,061.5   & \tbf{353.9}    & 707.7   &  707.7   & 353.8   & \tbf{1} & 1 \\
\texttt{kernel}      & 160 & 246   & \tbf{984.1}   & 2,952.1   & \tbf{984.1}    & 1,968.1 &  1,968.1 & 984.0   & \tbf{1} & 1 \\
\texttt{fib41}       & 2   & 256   & \tbf{1,022.0} & 3,066.0   & \tbf{1,022.0}  & 2,044.0 &  2,044.0 & 1,022.0 & \tbf{1} & 1 \\
\texttt{cere}        & 5   & 440   & \tbf{1,759.9} & 5,279.0   & \tbf{1,759.9}  & 3,519.3 &  3,519.3 & 1,759.7 & \tbf{1} & 1 \\ \hdashline
\texttt{bbba}         & 2   & 100   & 1,200.0       & 1,200.0   & 1,200.0        & 800.0   &  800.0   & 400.0   & \tbf{1} & 1 \\ \hline
\end{tabular}
\caption{
Working space in MB, except for \ourD and \sacak, which used 1 KB.
}
\label{t:workspace}
\end{table}}

We compared the performance of our algorithm, {called \ourA,} with algorithms to compute \LA in linear time 
by Franek \etal~\cite{Franek2016,Hohlweg2003} (\nsv), 
Baier~\cite{Baier2016,FranekPS17} (\baier), and
Louza \etal~\cite{jda/LouzaSMT18} (\louza).
We also compared a version of Baier's algorithm that computes \LA and \SA together (\baiersa). We considered the three linear time alternatives of our algorithm described in Sections~\ref{sec:alt1}--\ref{s:alt_3}. We tested all three versions since one could be interested in the fastest algorithm regardless of the space usage. We used four bytes for each computer word so the total space usage of our algorithms was respectively $17n$, $13n$ and $9n$ bytes. We included also the performance of \sacak~\cite{tois/Nong13} to evaluate the overhead added by the computation of \LA in addition to the \SA.

The experiments were conducted on a machine with an \texttt{Intel Xeon} Processor \texttt{E5-2630} v3 20M Cache 2.40-GHz, 384 GB of internal memory and a 13 TB SATA storage, under a 64 bits \texttt{Debian GNU/Linux 8} (kernel 3.16.0-4) OS.
We implemented our algorithms in ANSI C.
The time was measured with \texttt{clock()} function of C standard libraries and the memory was measured using \texttt{malloc\_count} library\footnote{\url{https://github.com/bingmann/malloc\_count}}.
The source-code is publicly available at \url{https://github.com/felipelouza/lyndon-array/}.

We used string collections from the Pizza~\&~Chili dataset\footnote{\url{http://pizzachili.dcc.uchile.cl/texts.html}}.
In particular, the datasets \texttt{einstein-de}, \texttt{kernel}, \texttt{fib41} and \texttt{cere} are highly repetitive texts\footnote{\url{http://pizzachili.dcc.uchile.cl/repcorpus.html}}, and
the \texttt{english.1G} is the first 1GB of the original \texttt{english} dataset. We also created an artificial repetitive dataset, called \texttt{bbba}, consisting of a string $T$ with $100\times2^{20}$ copies of $b$ followed by one occurrence of $a$, that is, $T=b^{n-2}a\$$. This dataset represents a worst-case input for the algorithms that use a stack (\nsv and \louza).

Table~\ref{t:time} shows the running time of each algorithm in $\mu$s/input byte.
The results show that our algorithm is competitive in practice. 
In particular, the version \ourD was only about $1.35$ times slower than the fastest algorithm (\baier) for non-repetitive datasets, and $2.92$ times slower for repetitive datasets.
Also, the performance of \ourD and \baiersa were very similar. Finally, the overhead of computing \LA in addition to \SA was small: \ourD was $1.42$ times slower than \sacak, whereas \baiersa was $1.55$ times slower than \baier, on average. Note that \ourD was consistently faster than \ourC and \ourB, so using more space does not yield any advantage. 

Table~\ref{t:peakspace} shows the peak space consumed by each algorithm given in bytes per input symbol.
The smallest values were obtained by \nsv, \louza and \ourD. 
In details, the space used by \nsv and \louza was $9n$ bytes plus the space used by the stack. 
The stack space was negligible (about 10KB) for almost all datasets, except for \tbf{bbba} where the stack used $4n$ bytes for \nsv and $8n$ bytes for \louza 
(the number of stack  entries is the same, but each stack entry consists of a pair of integers).
On the other hand, our algorithm, \ourD, used exactly $9n+1024$ bytes for all datasets.

\ignore{Table~\ref{t:workspace} shows the working space of each algorithm in MB, except for algorithms 
\ourD and \sacak, whose working space was exactly equal to $1$ KB.
The working space was obtained by subtracting from the total space the space used by the input string $T$ and the output. We considered three different scenarios. First, we considered only \LA as output, then we considered \LA and \SA, and finally only \SA.
The smallest values were given by \ourD, which was the only solution that kept the working space constant.
The additional space used \nsv and \louza were equal to the version \ourC, except for \tbf{bbba}.
In particular, for dataset \texttt{bbba}, we can see that solutions that compute \LA and \SA use less additional space than solutions that obtain only \LA.}

\section{Conclusions}

{We have introduced an algorithm for computing simultaneously the suffix array and Lyndon array (\LA) of a text using induced suffix sorting. 
The most space-economical variant of our algorithm uses only $n + \sigma + O(1)$ words of working space making it the most space economical \LA algorithm among the ones running in linear time; this includes both the algorithm computing the \SA and \LA and the ones computing only the \LA. The experiments have shown our algorithm is only slightly slower than the available alternatives, and 
that computing the \SA is usually the most expensive step of all linear time \LA construction algorithms.  
A natural open problem is to devise a linear time algorithm to construct only the \LA using $o(n)$ words of working space.}

\section*{Acknowledgments}

The authors thank Uwe Baier for kindly providing the source codes of algorithms \baier and \baiersa, and Prof. Nalvo Almeida for granting access to the machine used for the experiments.

\paragraph{Funding:}
F.A.L. was supported by the grant $\#$2017/09105-0 from the S\~ao Paulo Research Foundation (FAPESP).
G.M. was partially supported by PRIN grant 2017WR7SHH, by INdAM-GNCS Project 2019 {\sl Innovative methods for the solution of medical and biological big data} and by the  LSBC\_19-21 Project from the University of Eastern Piedmont.
S.M. and M.S. are partially supported by MIUR-SIR project CMACBioSeq {\sl Combinatorial methods for analysis and compression of biological sequences} grant n.~RBSI146R5L.
G.P.T. acknowledges the support of Brazilian agencies Conselho Nacional de Desenvolvimento Científico e Tecnológico (CNPq) and Coordenação de Aperfeiçoamento de Pessoal de Nível Superior (CAPES).

%
%


\end{document}